\documentclass[review]{elsarticle}

\usepackage{lineno,hyperref}
\usepackage{amssymb}
\usepackage{amsmath}
\usepackage{amsthm}

\usepackage{tikz}
\usetikzlibrary{calc,math}
\usetikzlibrary{arrows, arrows.meta}

\usepackage{algorithm}
\usepackage[noend]{algpseudocode}
\algrenewcommand\algorithmicrequire{\textbf{Input:}}
\algrenewcommand\algorithmicensure{\textbf{Output:}}
\algnewcommand{\algorithmicand}{\textbf{ and }}
\algnewcommand{\algorithmicor}{\textbf{ or }}
\algnewcommand{\OR}{\algorithmicor}
\algnewcommand{\AND}{\algorithmicand}

\newtheorem{theorem}{Theorem}[section]

\newtheorem{definition}[theorem]{Definition}

\newcommand{\ProblemA}{(1-SPNI)}
\newcommand{\Problem}{(2-SPNI)}

\journal{Journal of ...}

\biboptions{authoryear}

\begin{document}
	
	\begin{frontmatter}
		
		\title{The two player shortest path network interdiction problem}

		\author[a]{Simon Busam}
		\ead{simon.busam@web.de}
		\author[a]{Luca E. Sch\"afer\corref{mycorrespondingauthor}}
		\cortext[mycorrespondingauthor]{Corresponding author}
		\ead{luca.schaefer@mathematik.uni-kl.de}
		
		\author[a]{Stefan Ruzika}
		\ead{ruzika@mathematik.uni-kl.de}
		
		\address[a]{Department of Mathematics, Technische Universit\"at Kaiserslautern, 67663 Kaiserslautern, Germany}
		
		\begin{abstract}
			In this article, we study a biobjective extension of the shortest path network interdiction problem.  Each arc in the network is associated with two integer length values and two players compute their  respective shortest paths from source to sink independently from each other while an interdictor tries to lengthen both shortest paths by removing arcs. We show that this problem is intractable and that deciding whether a feasible interdiction strategy is efficient, is \(\mathcal{NP}\)-complete. We provide a solution procedure to solve the problem on two-terminal series-parallel graphs in pseudopolynomial time.
		\end{abstract}
		
		
	\end{frontmatter}
	
	\section{Introduction}
	The shortest path network interdiction problem \ProblemA\ usually involves two parties competing against each other. One player tries to compute its shortest path from source to sink, while the second player, called the interdictor, who is subject to a restricted interdiction budget, removes arcs from the network to maximally deteriorate the first players shortest path length. 
	
	One of the earliest works representing a special case of \ProblemA, called the \(k\) most vital arcs problem, in which the interdiction of an arc requires exactly one unit of the interdictors budget, has been studied by the authors in \cite{malik1989k}. 
	The \(k\) most vital arcs problem as well as \ProblemA\ have been shown to be \(\mathcal{NP}\)-hard, cf. \cite{bar1998complexity, ball1989finding}. In \cite{corley1982most}, the \(k\) most vital arcs problem is analyzed and related to the $k$ shortest path problem. The authors also provide an algorithm to obtain a most vital link, which is again a special case of the \(k\) most vital arcs problem for \(k\) equals \(1\). Instead of removing arcs, \cite{fulkerson1977maximizing} and \cite{golden1978problem} study related problems, where each arc is associated with an interdiction cost per unit to increase the effective length of that arc. Another variant is studied in \cite{khachiyan2008short}, where each vertex is associated with a number denoting how many outgoing arcs might be deleted. An extension of Dijkstra's algorithm efficiently solving the problem is provided along with inapproximability bounds for the $k$ most vital arcs problem and various related problems. This variant is in turn altered in \cite{andersson2009perfect} to solve a shortest path interdiction problem with node-wise budget, where partial interdiction is allowed. Additionally, a shortest path interdiction problem with node-wise budget and a bottleneck objective is considered along with an algorithmic idea for the shortest path interdiction problem with bottleneck objective, a global budget and unit interdiction costs. In one of the most prominent works dealing with \ProblemA\ in its most general form, cf. \cite{israeli2002shortest}, two algorithms with different quality, depending on whether interdicted arcs are removed or if an interdicted arc's length is increased by some value, are provided.
	
	However, literature on biobjective extensions and variants of \ProblemA\ is rather sparse.
	In \cite{ramirez2010bi}, a biobjective variant considering the maximization of the shortest path length and the minimization of interdiction costs is investigated. Finding an optimal route for an ambulance is considered in \cite{torchiani2017shortest}. The problem is modelled as a biobjective problem where the first objective seeks to minimize the shortest path from source to sink while the second objective minimizes the maximal length of a detour in case the chosen route is blocked. 
	
	\paragraph{Our contribution}
	To the best of our knowledge, we provide the first extension of \ProblemA\ involving an additional player, called \Problem: Each arc in the network is associated with two integer lengths and both the first and the second player aim to compute their respective shortest path from source to sink. The interdictor's task is to remove arcs from the network while satisfying a given interdiction budget to maximize both the first and the second players objective. This could be of particular interest whenever two different parties want to pass through a common network as fast as possible, for which one aims to identify the most critical components, see for example nuclear smuggling interdiction, cf. \cite{morton2007models}.
	We formally introduce the problem in Section \ref{sec:problem} and prove that the number of non-dominated points might be exponential in the number of vertices of the network, see Section \ref{sec:complexity}. Additionally, we show that deciding whether a feasible interdiction strategy is efficient or not, is \(\mathcal{NP}\)-complete. In Section \ref{sec:spgraphs}, we discuss a pseudopolynomial time dynamic programming algorithm on two-terminal series-parallel graphs for \Problem.
	Section \ref{sec:conclusion} summarizes the article and provides further directions of research.
	
	\section{Preliminaries and problem formulation}\label{sec:problem}
	Let \(G=(V,A)\) be a directed network with vertex set~$V$ and arc set~$A$ with \(n:=|V|\) and \(m:=|A|\). If the network \(G\) is not clear from the context, we write \(V(G)\) and \(A(G)\) to refer to the set of vertices and arcs of \(G\), respectively.
	Let \(s,t\in V\) be the source and sink vertex in \(G\), respectively. 
	Each arc in \(G\) is associated with two integer length values, i.e., \(l=(l^1,l^2):A\rightarrow \mathbb{N}^2\) with \(l^i: A\rightarrow \mathbb{N}\) for \(i=1,2\). The maximum arc length is denoted by \(L^1\) and \(L^2\) for player one and two, respectively, i.e., \(L^i:=\max\{l^i(a)\mid a \in A\}\). By \(L_{\max}\), we denote the maximum of \(L^1\) and \(L^2\).
	Note that two possibly different shortest \(s\)-\(t\)-paths \(P^1\) and \(P^2\) might be computed in \(G\) with respect to \(l^1\) and \(l^2\), respectively. By \(\mathcal{P}_{st}(G)\), we denote the set of all \(s\)-\(t\)-paths in \(G\) and we compute the length of a path \(P\in\mathcal{P}_{st}(G)\) with respect to \(l^i\) as the sum of all arc lengths in \(P\), i.e., \(l^i(P):=\sum_{a \in A(P)} l^i(a)\), where \(A(P)\) denotes the set of arcs in \(P\).
	Further, interdicting arc \(a \in A\) is associated with a cost \(c(a)\in \mathbb{N}\) and \(B\in \mathbb{N}\) denotes the given interdiction budget. Thus, by \(\Gamma\) we refer to the set of all feasible interdiction strategies, i.e.,  
	\begin{equation*}
	\Gamma := \left\{\gamma = (\gamma_a)_{a \in A} \in \{0,1\}^m \mid \sum\limits_{a \in A} c(a)\cdot\gamma_a \leq B\right\},	
	\end{equation*} 
	where the binary variable \(\gamma_a\) equals one if arc \(a\) is interdicted and zero otherwise. Consequently, each feasible interdiction strategy \(\gamma \in \Gamma\) induces an interdicted graph \(G(\gamma):= (V', A')\) with \(V'=V\) and \(A' = A\setminus A(\gamma)\), where \(A(\gamma) := \left\{a \in A \mid \gamma_a = 1 \right\}\). 
	
	We interpret this setting as a game composed of three players. 
	Whereas player one and two compute their respective shortest \(s\)-\(t\)-paths independently from each other in \(G(\gamma)\) for some \(\gamma\in\Gamma\), the interdictor aims to maximize both shortest path lengths simultaneously by fixing some interdiction strategy \(\gamma\in\Gamma\). Thus, each interdiction strategy yields a tuple of shortest path lengths in \(G\), i.e., 
	\begin{equation*}
	f_G: \Gamma \rightarrow \mathbb{N}^2,\quad \gamma \rightarrow \begin{pmatrix}\min\limits_{P\in\mathcal{P}_{st}(G(\gamma))} l^1(P) \\ \min\limits_{P\in\mathcal{P}_{st}(G(\gamma))} l^2(P)\end{pmatrix} := \begin{pmatrix} f^1_G(\gamma) \\ f^2_G(\gamma)\end{pmatrix}.
	\end{equation*}
	To compare vectors of shortest path lengths, we use the standard  Pareto-order in \(\mathbb{N}^2\), cf. \cite{ehrgott2005multicriteria}, which is defined as follows:
	\begin{equation*}
	y_1 \geq y_2 \Leftrightarrow y_1^k \geq y_2^k \text{ for } k = 1,2 \text{ and } y_1 \neq y_2,
	\end{equation*}
	where \(y_1=(y_1^1,y_1^2)\) and \(y_2=(y_2^1,y_2^2)\).
	Since \(\geq\) does not define a total order on the objective function values in \(\mathbb{N}^2\), one aims to find the feasible interdiction strategies \(\gamma\in\Gamma\) that do not allow to improve the objective of the first shortest path player without deteriorating the second player's objective.
	Thus, \Problem\ can be stated as \(\max_{\gamma\in\Gamma} f_G(\gamma)\).
	\begin{definition}
		A feasible interdiction strategy \(\gamma \in \Gamma\) is called efficient, if there does not exist \(\gamma' \in \Gamma\) such that
		\begin{equation*}
		f_G(\gamma') \geq f_G(\gamma). 
		\end{equation*}
		In this case, we call \(f_G(\gamma)\) a non-dominated point. With \(\Gamma_E\), we denote the set of efficient interdiction strategies. The set of all non-dominated points is denoted by \(Z_N\).
	\end{definition}
	
	A special emphasis is put on two-terminal series-parallel graphs, which are defined as follows, cf. \cite{eppstein1992parallel}. 
	\begin{definition}\label{def:spgraph}
		A directed network \(G=(V,A)\) is called two-terminal series-parallel with source \(s\) and sink \(t\), if \(G\) can be constructed by a sequence of the following operations.
		\begin{itemize}
			\item[1)] Construct a primitive graph \(G'=(V',A')\) with \(V'=\{s,t\}\) and \(A'=\{(s,t)\}\).
			\item[2)] (Parallel Composition) Given two directed, series-parallel graphs \(G_1\) with source \(s_1\) and sink \(t_1\) and \(G_2\) with source \(s_2\) and sink \(t_2\), form a new graph \(G\) by identifying \(s=s_1=s_2\) and \(t=t_1=t_2\).
			\item[3)] (Series Composition) Given two directed, series-parallel graphs \(G_1\) with source \(s_1\) and sink \(t_1\) and \(G_2\) with source \(s_2\) and sink \(t_2\), form a new graph \(G\) by identifying \(s=s_1\), \(t_1=s_2\) and \(t_2=t\).
		\end{itemize}
	\end{definition}
	
	A two-terminal series-parallel graph \(G\) can be recognized in polynomial time along with the corresponding decomposition tree \(T_G\). Further, the size of \(T_G\) is linear in the size of \(G\) and \(T_G\) can be computed in linear time, cf. \cite{valdes1982recognition}. The decomposition tree \(T_G\) specifies how \(G\) has been constructed using the rules mentioned in Definition \ref{def:spgraph}. In the following, we assume that each vertex in \(T_G\) is associated with a two-terminal series-parallel graph, for an example see Figure \ref{fig:decomposition_tree}. Thus, if we refer to a graph \(H\) in \(T_G\), we actually refer to the graph \(H\), which actually denotes a subgraph of \(G\), corresponding to a vertex in \(T_G\).
	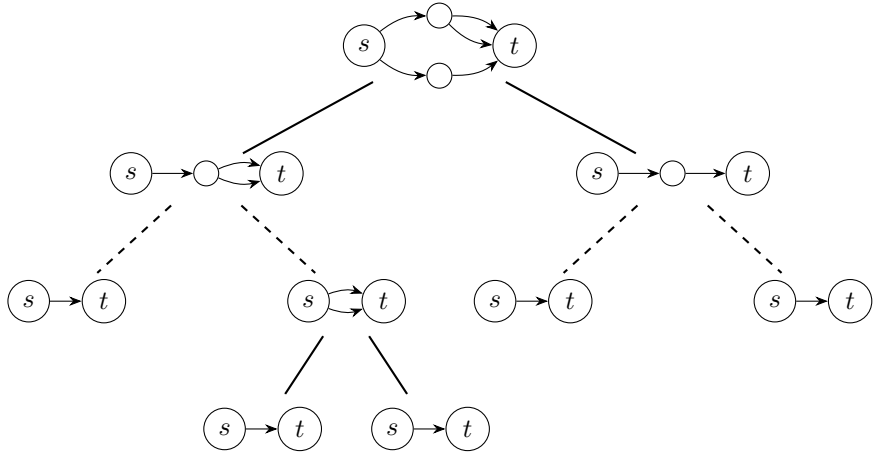
\begin{figure}
		\begin{center}
			\begin{tikzpicture}
			\tikzmath{
				\smallx = 1;
				\smally = 0.4;
				\basex = 3.1;
				\basey = 1.7;
				\factor = 0.6;
				\gap = 0.0;
				\bigBend = 40;
				\smallBend = 20;
				\rootx = 0;
				\rooty = 0;
				\rootlx = \rootx - \basex;
				\rootrx = \rootx + \basex;
				\leveloney = \rooty - \basey;
				\rootllx = \rootlx - \factor*\basex;
				\rootlrx = \rootlx + \factor*\basex;
				\rootrlx = \rootrx - \factor*\basex;
				\rootrrx = \rootrx + \factor*\basex;
				\leveltwoy = \rooty - 2*\basey;
				\rootlrlx = \rootlrx - \factor*\factor*\basex;
				\rootlrrx = \rootlrx + \factor*\factor*\basex;
				\levelthreey = \rooty - 3*\basey;
				\rightshift = 0.0*\smallx;
				\radius = 0.55;
			}
			
			\draw[thick] (\rootx-\gap+2*\rightshift, \rooty-\gap) -- (\rootlx+\gap+2*\rightshift, \leveloney+\gap);
			\draw[thick] (\rootx+\gap+2*\rightshift, \rooty-\gap) -- (\rootrx-\gap+2*\rightshift, \leveloney+\gap);
			\draw[thick, dashed] (\rootlx-\gap+2*\rightshift, \leveloney-\gap) -- (\rootllx+\gap+1*\rightshift, \leveltwoy+\gap);
			\draw[thick, dashed] (\rootlx+\gap+2*\rightshift, \leveloney-\gap) -- (\rootlrx-\gap+1*\rightshift, \leveltwoy+\gap);
			\draw[thick] (\rootlrx-\gap+1*\rightshift, \leveltwoy-\gap) -- (\rootlrlx+\gap+1*\rightshift, \levelthreey+\gap);
			\draw[thick] (\rootlrx+\gap+1*\rightshift, \leveltwoy-\gap) -- (\rootlrrx-\gap+1*\rightshift, \levelthreey+\gap);
			\draw[thick, dashed] (\rootrx-\gap+2*\rightshift, \leveloney-\gap) -- (\rootrlx+\gap+1*\rightshift, \leveltwoy+\gap);
			\draw[thick, dashed] (\rootrx+\gap+2*\rightshift, \leveloney-\gap) -- (\rootrrx-\gap+1*\rightshift, \leveltwoy+\gap);
			
			\filldraw [white] (\rootx,\rooty) circle (1);
			\filldraw [white] (\rootlx,\leveloney) circle (\radius);
			\filldraw [white] (\rootllx,\leveltwoy) circle (\radius);
			\filldraw [white] (\rootlrx,\leveltwoy) circle (\radius);
			\filldraw [white] (\rootlrlx,\levelthreey) circle (\radius);
			\filldraw [white] (\rootlrrx,\levelthreey) circle (\radius);
			\filldraw [white] (\rootrx,\leveloney) circle (\radius);
			\filldraw [white] (\rootrlx,\leveltwoy) circle (\radius);
			\filldraw [white] (\rootrrx,\leveltwoy) circle (\radius);
			
			\begin{scope}[every node/.style={circle,draw, fill=white}]
			\node (S) at  ( \rootx-1*\smallx, 0) {$s$};
			\node (v1) at  ( \rootx, \smally) {};
			\node (v2) at  ( \rootx, -\smally) {};
			\node (T) at  ( \rootx+1*\smallx, 0) {$t$};
			\end{scope}
			\begin{scope}[>={Stealth[black]}, sloped]
			\path [->] (S) [bend left=\smallBend]edge node [below] {} (v1);
			\path [->] (S) [bend right=\smallBend]edge node [below] {} (v2);
			\path [->] (v2) [bend right=\smallBend]edge node [below] {} (T);
			
			\path [->] (v1) [bend left=\smallBend] edge node [above] {} (T);
			\path [->] (v1) [bend right=\smallBend] edge node [above] {} (T);
			\end{scope}
			
			\begin{scope}[every node/.style={circle,draw, fill=white}]
			\node (S) at  (-1*\smallx+\rootlx, 0+\leveloney) {$s$};
			\node (v) at  ( 0*\smallx+\rootlx, 0+\leveloney) {};
			\node (T) at  ( 1*\smallx+\rootlx, 0+\leveloney) {$t$};
			\end{scope}
			\begin{scope}[>={Stealth[black]}, sloped]
			\path [->] (S) [bend left=0]edge node [below] {} (v);
			\path [->] (v) [bend left=\smallBend] edge node [above] {} (T);
			\path [->] (v) [bend right=\smallBend] edge node [above] {} (T);
			\end{scope}
			
			\begin{scope}[every node/.style={circle,draw, fill=white}]
			\node (S) at  (-0.5*\smallx+\rootllx, 0+\leveltwoy) {$s$};
			\node (T) at  ( 0.5*\smallx+\rootllx, 0+\leveltwoy) {$t$};
			\end{scope}
			\begin{scope}[>={Stealth[black]}, sloped]
			\path [->] (S) [bend right=0]edge node [below] {} (T);
			\end{scope}
			
			\begin{scope}[every node/.style={circle,draw, fill=white}]
			\node (S) at  (-0.5*\smallx+\rootlrx, 0+\leveltwoy) {$s$};
			\node (T) at  ( 0.5*\smallx+\rootlrx, 0+\leveltwoy) {$t$};
			\end{scope}
			\begin{scope}[>={Stealth[black]}, sloped]
			\path [->] (S) [bend right=\smallBend]edge node [below] {} (T);
			\path [->] (S) [bend left=\smallBend]edge node [below] {} (T);
			\end{scope}
			
			\begin{scope}[every node/.style={circle,draw, fill=white}]
			\node (S) at  (-0.5*\smallx+\rootlrlx, 0+\levelthreey) {$s$};
			\node (T) at  ( 0.5*\smallx+\rootlrlx, 0+\levelthreey) {$t$};
			\end{scope}
			\begin{scope}[>={Stealth[black]}, sloped]
			\path [->] (S) [bend right=0]edge node [below] {} (T);
			\end{scope}
			
			\begin{scope}[every node/.style={circle,draw, fill=white}]
			\node (S) at  (-0.5*\smallx+\rootlrrx, 0+\levelthreey) {$s$};
			\node (T) at  ( 0.5*\smallx+\rootlrrx, 0+\levelthreey) {$t$};
			\end{scope}
			\begin{scope}[>={Stealth[black]}, sloped]
			\path [->] (S) [bend right=0]edge node [below] {} (T);
			\end{scope}
			
			\begin{scope}[every node/.style={circle,draw, fill=white}]
			\node (S) at  (-1*\smallx+\rootrx, 0+\leveloney) {$s$};
			\node (v) at  ( 0*\smallx+\rootrx, 0+\leveloney) {};
			\node (T) at  ( 1*\smallx+\rootrx, 0+\leveloney) {$t$};
			\end{scope}
			\begin{scope}[>={Stealth[black]}, sloped]
			\path [->] (S) [bend left=0]edge node [below] {} (v);
			\path [->] (v) [bend left=0] edge node [above] {} (T);
			\end{scope}
			
			\begin{scope}[every node/.style={circle,draw, fill=white}]
			\node (S) at  (-0.5*\smallx+\rootrlx, 0+\leveltwoy) {$s$};
			\node (T) at  ( 0.5*\smallx+\rootrlx, 0+\leveltwoy) {$t$};
			\end{scope}
			\begin{scope}[>={Stealth[black]}, sloped]
			\path [->] (S) [bend right=0]edge node [below] {} (T);
			\end{scope}
			
			\begin{scope}[every node/.style={circle,draw, fill=white}]
			\node (S) at  (-0.5*\smallx+\rootrrx, 0+\leveltwoy) {$s$};
			\node (T) at  ( 0.5*\smallx+\rootrrx, 0+\leveltwoy) {$t$};
			\end{scope}
			\begin{scope}[>={Stealth[black]}, sloped]
			\path [->] (S) [bend right=0]edge node [below] {} (T);
			\end{scope}
			\end{tikzpicture}
		\end{center}
		\caption{Decomposition tree \(T_G\) of a two-terminal series-parallel graph $G$. The root vertex corresponds to $G$ itself. Every leaf corresponds to a primitive graph, i.e., a graph including only a single arc of \(G\). Dashed lines correspond to series compositions, whereas straight lines correspond to parallel compositions.}
		\label{fig:decomposition_tree}
	\end{figure}

	In the following, \((G,l,c,B)\) denotes an instance of \Problem, where \(G=(V,A)\) is a directed network, \(l\) denotes a length function assigning two lengths to each arc, \(c\) assigns every arc an interdiction cost and \(B\) refers to the interdiction budget.

	\section{Complexity}\label{sec:complexity}
	In this section, we investigate the complexity and tractability of \Problem. Therefore, we consider the following decision version of \Problem, which asks whether a given interdiction strategy is efficient or not: Given an instance \((G,l,c,B)\) of \Problem\ and a value \(K = (K^1,K^2) \in \mathbb{N}^2\), decide whether there exists an interdiction strategy \(\gamma \in \Gamma\) with \(f_G(\gamma)\geq K\).
	\begin{theorem}\label{thm:npc}
		The decision version of \Problem\ is \(\mathcal{NP}\)-complete, even on two-terminal series-parallel graphs.
	\end{theorem}
	\begin{proof}
		Follows immediately from the fact that the decision version of \ProblemA\ is \(\mathcal{NP}\)-complete on two-terminal series-parallel graphs, cf. \cite{ball1989finding}. Although the reduction proof presented by the authors does not use a two-terminal series-parallel graph, it can be easily modified such that \(\mathcal{NP}\)-completeness also holds for two-terminal series-parallel graphs.
	\end{proof}
	
	Further, in the field of multiobjective optimization one is usually interested in the (in)tractability of a specific problem, i.e., to investigate whether there exists a problem instance with an exponential number of non-dominated points with respect to the size of that instance.
	\begin{theorem}
		The problem \Problem\ is intractable, even on two-terminal series-parallel graphs and even for unit interdiction costs, i.e., the number of non-dominated points might be exponential in the size of the problem instance.
	\end{theorem}
	\begin{proof}
		To prove intractability of \Problem, we construct an instance, where the number of non-dominated points is exponential in the number of vertices. Therefore, let \((G,l,c,B)\) be an instance of \Problem\ with \(c(a)=1\) for all \(a \in A\) and \(n+2\) vertices, i.e., \(V:= \{s=v_0, v_1, \ldots,v_{n+1}=t\}\) with \(n\) being odd. There are two types of arcs going from \(v_i\) to \(v_{i+1}\) for all \(i=0,\ldots,n\), denoted by \(a^1_i\) and \(a^2_i\), respectively, with \(l(a^1_i)=(0,0)\) and \(l(a^2_i) = (2^i, 2^n-2^i)\). We create \(\frac{n+1}{2}\) copies of the latter type of arcs for all \(i=0,\ldots,n\) such that the resulting network has \((\frac{n+1}{2}+1)(n+1)\) arcs. Further, we set \(B=\frac{n+1}{2}\). Note that due to construction there does not exist a \(\gamma\in \Gamma\) such that we can separate \(s\) from \(t\). Further, one can see that all \(\gamma\in\Gamma\) with \(\sum_{a \in A}\gamma_a<B\) cannot be efficient and that it is always beneficial to interdict arcs of type \(a^1\) instead of \(a^2\).
		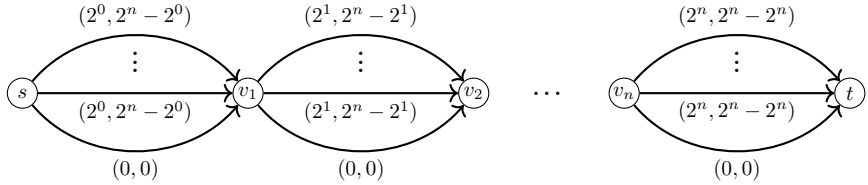
\begin{figure}[htbp]
			\centering
			\begin{tikzpicture}[scale=1, every node/.style={scale=0.8}]
			\node[draw, circle, minimum size=0.5cm, inner sep=0pt] (s) at (0,0) {\(s\)};
			\node[draw, circle, minimum size=0.5cm, inner sep=0pt] (a) at (3,0) {\(v_1\)};
			\node[draw, circle, minimum size=0.5cm, inner sep=0pt] (b) at (6,0) {\(v_2\)};
			\node[draw, circle, minimum size=0.5cm, inner sep=0pt] (c) at (8,0) {\(v_n\)};
			\node[draw, circle, minimum size=0.5cm, inner sep=0pt] (t) at (11,0) {\(t\)};
			
			\node[] (v) at (7,0){\textbf{\dots}};
			
			\draw [->,thick] (s)edge [bend angle=50, bend left] node[above]{$(2^0, 2^n-2^0)$}(a);
			\draw[->, thick] (s)--node[below]{$(2^0, 2^n-2^0)$}(a);
			\draw [->,thick] (s)edge [bend angle=50, bend right] node[below]{$(0,0)$}(a);
			\node[] (v) at (1.5,0.5){\textbf{\vdots}};
			\draw [->,thick] (a)edge [bend angle=50, bend left] node[above]{$(2^1, 2^n-2^1)$}(b);
			\draw[->, thick] (a)--node[below]{$(2^1, 2^n-2^1)$}(b);
			\draw [->,thick] (a)edge [bend angle=50, bend right] node[below]{$(0,0)$}(b);
			\node[] (v) at (4.5,0.5){\textbf{\vdots}};
			\draw [->,thick] (c)edge [bend angle=50, bend left] node[above]{$(2^{n}, 2^n-2^{n})$}(t);
			\draw[->, thick] (c)--node[below]{$(2^{n}, 2^n-2^{n})$}(t);
			\draw [->,thick] (c)edge [bend angle=50, bend right] node[below]{$(0,0)$}(t);
			\node[] (v) at (9.5,0.5){\textbf{\vdots}};
			\end{tikzpicture}
			\caption{Intractability instance of \Problem}
			\label{fig:intractable}
		\end{figure}
		Therefore, we only consider those interdiction strategies \(\gamma\in\Gamma\) with \(\sum_{a \in A}\gamma_a=B\) and \(\gamma_a\) equals \(0\) for all \(a\) with \(l(a)\neq(0,0)\). We denote the set by \(\Gamma^*\). It holds for the shortest path length of the first player that:
		\begin{equation*}
		2^{\frac{n+1}{2}}-1 \leq f^1_G(\gamma)\leq 2^{n+1} - 2^{\frac{n+1}{2}} \quad \text{ for all } \gamma \in \Gamma^*.
		\end{equation*}
		Note that the lower and upper bound are attained by removing the first and last \(B\) arcs of type \(a^1\), respectively. Further, for every \(\gamma \in \Gamma^*\),  it holds that \(f^2_G(\gamma) = B2^n - f^1_G(\gamma)\). Additionally,  \(f^1_G(\gamma)\neq f^1_G(\gamma')\) for all \(\gamma, \gamma'\in\Gamma^*\) with \(\gamma\neq\gamma'\). Thus, each \(\gamma\in\Gamma^*\) induces a different non-dominated point. Since \(|\Gamma^*|={n+1 \choose B} = {n+1 \choose \frac{n+1}{2}}\) and using Stirling's formula, we showed that the number of non-dominated points is exponential in \(n\), which concludes the proof.
	\end{proof}
	
	\section{Solution method for two-terminal series-parallel graphs}\label{sec:spgraphs}
	We state a dynamic programming algorithm with a pseudopolynomial running time for the case of two-terminal series-parallel graphs. 
	Throughout this section, we assume that a two-terminal series-parallel graph \(G\) is accompanied by its decomposition tree \(T_G\). We derive a dynamic programming algorithm starting at the leaves of \(T_G\) and iterating bottom up through the decomposition tree. In the course of the algorithm, we create labels of the form \((f^1_H(\gamma), f^2_H(\gamma))\) with \(\gamma\in\Gamma\) and \(H\) being a graph in \(T_G\). By \(\mathcal{L}(H, x)\), we denote the set of all labels correponding to non-dominated points in the graph \(H\) (induced by efficient interdiction strategies in \(H\)) with a total interdiction cost of \(x\), i.e., when the interdictor's budget is exactly \(x\). We aim to find the set of all non-dominated points for all graphs in \(T_G\) and for all \(x\in\{0,\ldots,B\}\).
	
	If \(H=(V_H,A_H)\) is a primitive graph, i.e., a leaf of \(T_G\), with \(V_H=\{s_H,t_H\}\) and \(A_H=\{a^*=(s_H, t_H)\}\), we can clearly calculate \(\mathcal{L}(H,x)\) in case of \(c(a^*)\leq B\) for all \(x \in \{0,1,\ldots,B\}\), in the following way:
	\begin{equation}\label{eq:1}
	\mathcal{L}(H,x) =\begin{cases}
	\{(l^1(a^*),l^2(a^*))\}, & \text{ if } x = 0,1,\ldots,c(a^*)-1\\
	\{(\infty,\infty)\}, & \text{ if } x = c(a^*),\ldots,B
	\end{cases}
	\end{equation}
	If \(c(a^*)>B\), then \(\mathcal{L}(H,x)\) is equal to \(\{(l^1(a^*),l^2(a^*))\}\) for all \(x\in\{0,1,\ldots,B\}\).
	
	For \(H\) being the series or parallel composition of \(H_1\) and \(H_2\), we define the following two operations.
	\begin{definition}
		Let \(A,B\subset \mathbb{N}^2\) be two sets. Then, 
		\begin{itemize}
			\item \(A\oplus B := \{a+b\mid a\in A, b\in B\}\) \Comment{Minkowski sum}
			\item  \(A\odot B:=\{a\odot b\mid a\in A, b\in B\}\), where \(a\odot b:=(\min\{a^1, b^1\}, \min\{a^2, b^2\})\) with \(a=(a^1,a^2)\) and \(b=(b^1,b^2)\)
		\end{itemize}
		
	\end{definition}
	
	If \(H\) is the parallel composition of \(H_1\) and \(H_2\), we calculate \(\mathcal{L}(H, x)\) as follows:
	\begin{equation}\label{eq:3}
	\mathcal{L}(H,x) = \max \left\{\bigcup_{k=0}^x \mathcal{L}(H_1,k) \odot \mathcal{L}(H_2,x-k)\right\} \text{ for } x = 0,1,\ldots,B.
	\end{equation}
	Thus, we combine each non-dominated point \(\mathcal{L}(H_1,k)\) with each non-dominated point in \(\mathcal{L}(H_2,x-k)\) by taking the respective minimum in each component. By \(`\max`\), we denote that all dominated points get discarded afterwards with respect to the Pareto-order.
	
	If \(H\) is the series composition of \(H_1\) and \(H_2\), we calculate \(\mathcal{L}(H, x)\) as follows:
	\begin{equation}\label{eq:2}
	\mathcal{L}(H,x) = \max \left\{\bigcup_{k=0}^x \mathcal{L}(H_1,k) \oplus \mathcal{L}(H_2,x-k)\right\} \text{ for } x = 0,1,\ldots,B.
	\end{equation}
	In this case, non-dominated points are combined by summing them up. Again, dominated points get discarded afterwards.

	\begin{theorem}\label{thm:correctness}
		The set of non-dominated points of \Problem\ can be computed on two-terminal series-parallel graphs by using formulas \eqref{eq:1}, \eqref{eq:3} and \eqref{eq:2}.
	\end{theorem}
	\begin{proof}
		We use induction on the size of the decomposition tree \(T_G\) of \(G\). Let \(H\) be a graph in \(T_G\).
		If \(H\) is a primitive graph, i.e., a leaf vertex of \(T_G\), the set of non-dominated points can clearly be computed by using \eqref{eq:1}.
		Now, let \(H\) be the series composition of \(H_1\) and \(H_2\). For the sake of contradiction, suppose \(y=(f^1_H(\gamma), f^2_H(\gamma))\) is a non-dominated point in \(\mathcal{L}(H,x^*)\) for some \(\gamma\in\Gamma\) and \(x^* \in \{0,\ldots,B\}\) that has not been found. Let \(y=p+q\), where \(p=(f^1_{H_1}(\gamma^1), f^2_{H_1}(\gamma^1))\) and \(q=(f^1_{H_2}(\gamma^2), f^2_{H_2}(\gamma^2))\) with \(\gamma=\gamma^1+\gamma^2\). Let \(c = \sum_{a\in A}c(a)\cdot\gamma^1_a\). 
		If \(p~\in~\mathcal{L}(H_1,c)\) and \(q\in \mathcal{L}(H_2,x^*-c)\), then \(y\) would have been created at \(\mathcal{L}(H,x^*)\) due to construction of the algorithm. 
		Thus, we assume that \(p\notin \mathcal{L}(H_1,c)\) or \(q\notin \mathcal{L}(H_2,x^*-c)\). Without loss of generality, we assume that \(p\notin \mathcal{L}(H_1,c)\). It follows that there exists a non-dominated point \(r \in \mathcal{L}(H_1,c)\) with \(r\geq p\). Consequently, it holds that \(r+q\geq y\), which is a contradiction to our assumption that \(y\) is non-dominated. 
		The claim can analogously be proven for the case of \(H\) being the parallel composition of \(H_1\) and \(H_2\).
	\end{proof}

	Further, the above described dynamic programming algorithm runs in pseudo-polynomial time.
	\begin{theorem}
		The dynamic programming algorithm has a worst-case running-time complexity of \(\mathcal{O}(mB^2n^2L_{\max}^2\log(BnL_{\max}))\).
	\end{theorem}
	\begin{proof}
		The decomposition tree \(T_G\) has \(m\) leaf vertices. Since it is a binary tree, we know that \(T_G\) has \(2m-1\) vertices. Each leaf requires constant time for determining one label set. For each leaf we construct \(B+1\) labels. 
		Let \(H\) be a subgraph of \(G\) corresponding to one of the \(m-1\) non-leaves in \(T_G\). 
		We call $H_1$ and $H_2$ the subgraphs of $H$ corresponding to the children of \(H\) in the decomposition tree. 
		Every path in \(G\), independent of the interdiction strategy, can be of length between \(0\) and \((n-1)L_{\max}\) or \(\infty\) for both players. 
		Thus, the number of non-dominated labels in \(\mathcal{L}(H,x)\) is in \(\mathcal{O}(nL_{\max})\) for all \(x\in\{0,1,\ldots,B\}\). For each of the \(m-1\) non-leaves we require two steps. First, we create \(\sum_{x=0}^{B}\sum_{k=0}^{x}|\mathcal{L}(H_1,k)|\cdot|\mathcal{L}(H_2,x-k)|\in\mathcal{O}(B^2n^2L_{\max}^2)\) labels. Second, we have to check the labels for non-dominance, which can be done in \(\mathcal{O}(B^2n^2L_{\max}^2\log(BnL_{\max}))\), cf. \cite{kung1975finding}. Executing these operations at most \(m-1\) times, yields an overall running time complexity of \(\mathcal{O}(mB^2n^2L_{\max}^2\log(BnL_{\max}))\), which is pseudopolynomial in the size of the input.
	\end{proof}

	\section{Conclusion}\label{sec:conclusion}
	In this article, we provided a biobjective extension of the shortest path network interdiction problem resulting in a game composed of three players, i.e., two shortest path players and one interdictor.
	The discussed decision version of the two player shortest path network interdiction problem was shown to be \(\mathcal{NP}\)-complete, even for two-terminal series-parallel graphs. Further, we provided an instance with an exponential number of non-dominated points proving the problem to be intractable, even for unit interdiction costs and on two-terminal series-parallel graphs. We provided a dynamic programming algorithm solving the problem on two-terminal series-parallel graphs in pseudopolynomial time.
	
	The two player shortest path network interdiction problem is open to a wide range of further research, including potential approximation and/or pseudo-approximation approaches. The development of algorithms for more general graph classes provides further interesting fields of study. 
	
	\section*{Acknowledgments}
	This work was partially supported by the Bundesministerium f\"ur Bildung und Forschung (BMBF) under Grant No. 13N14561 and Deutscher Akademischer Austauschdienst (DAAD) under Grant No. 57518713.
	
%

\end{document}